\documentclass[11pt]{article}

\usepackage{amsmath, amsthm, amsfonts}
\usepackage[margin=2.5cm]{geometry}
\usepackage{amssymb}
\usepackage{fancyhdr}
\usepackage[cp1250]{inputenc}

\newtheorem{thm}{Theorem}[section]
\newtheorem{rem}{Remark}[section]

\newtheorem{defn}{Definition}[section]

\usepackage{hyperref}
\hypersetup{
    colorlinks=true,
    linkcolor=blue,
    filecolor=magenta,      
    urlcolor=cyan,
}

\title{ Applications of the Stroboscopic Tomography to Selected 2-Level Decoherence Models}
\author{Artur Czerwi{\'n}ski\\
\small ResearchGate: \url{www.researchgate.net/profile/Artur_Czerwinski} \\
\small Link to the article in Journal page: \url{http://link.springer.com/article/10.1007/s10773-015-2703-2}\\
  \small 1. Institute of Physics\\
 \small Nicolaus Copernicus University\\
  \small 87-100 Toru{\'n}\\
	\small 2. Center for Theoretical Physics \\
\small Polish Academy of Sciences \\
\small 02-668 Warszawa\\
\small DOI: 10.1007/s10773-015-2703-2 \\
}
\begin{document}
\maketitle

\begin{abstract}
In the paper we discuss possible applications of the so-called stroboscopic tomography (stroboscopic observability) to selected decoherence models of 2-level quantum systems. The main assumption behind our reasoning claims that the time evolution of the analyzed system is given by a master equation of the form $\dot{\rho} = \mathbb{L} \rho$ and the macroscopic information about the system is provided by the mean values $m_i (t_j) = Tr(Q_i \rho(t_j))$ of certain observables $\{Q_i\}_{i=1} ^r $ measured at different time instants $\{t_j\}_{j=1}^p$. The goal of the stroboscopic tomography is to establish the optimal criteria for observability of a quantum system, i.e. minimal value of $r$ and $p$ as well as the properties of the observables $\{Q_i\}_{i=1} ^r $.

\end{abstract}

\section{Introduction}

According to one of the most fundamental assumptions of quantum theory, the density matrix carries the achievable information about the quantum state of a physical system. In recent years the determination of the trajectory of the state based on the results of measurements has gained new relevance because the ability to create, control and manipulate quantum states has found applications in other areas of science, such as: quantum information theory, quantum communication and computing.

The identification of an unknown state by appropriate measurements is possible only if we have a set of identical copies of this state, because each state can be measured only once due to the fact that every measurement, in general, changes the state. Moreover, in order to create a successful model of quantum tomography one needs to find a collection of observables, such that their mean values provide the complete information about the state. In the standard approach to quantum tomography of 2-level systems one takes as the observables the set of Pauli matrices, denoted by $\{\sigma_1, \sigma_2, \sigma_3\}$, see for example \cite{altepeter04}. The reconstruction of the initial density matrix is possible due to the decomposition in the basis $\{ \mathbb{I}, \sigma_1, \sigma_2, \sigma_3 \}$, which has the form
\begin{equation}
\rho (0) = \frac{1}{2} \left( \mathbb{I} + \sum_{i=1}^3 s_i \sigma_i \right),
\end{equation}
where $s_i$ is the expectation value of $\sigma_i$ in state $\rho(0)$.
Thus in the standard approach one needs to measure three different physical quantities in order to reconstruct the density matrix of a 2-level system. In general for an N-level system one would need to measure $N^2 - 1$ different observables - more about general approach can be found in \cite{alicki87,genki03}. This fact implies that the standard approach seems rather impracticable as from experimental point of view it is difficult to find as many different observables.

Therefore, in this paper we follow the stroboscopic approach to quantum tomography which was proposed in \cite{jam83} and then developed in \cite{jam00,jam04}. In the stroboscopic approach we consider a set of observables $\{Q_i\}_{i=1} ^r $ (where $r< N^2 -1$) and each of them can be measured at time instants $\{t_j\}_{j=1}^p$. Every measurement provides a result that shall be denoted by $m_i (t_j)$ and can be represented as $m_i (t_j) = Tr(Q_i \rho(t_j))$. Because in this approach the measurements are performed at different time instants, it is necessary to assume that the knowledge about the character of evolution is available, e.g. the Kossakowski-Lindblad master equation \cite{gorini76} is known or, equivalently, the collection of Kraus operators. Knowledge about the evolution makes it possible to determine not only the initial density matrix but also the complete trajectory of the state. To make this issue clearer from now on we assume the following definition \cite{jam04}.
\begin{defn}
An N-level open quantum system is said to be $(Q_1, ...,Q_r)$-reconstructible on an interval $[0,T]$ if there exists at least one set of time instants $\{t_j\}_{j=1}^p$ ordered as $0\leq t_1 < ...< t_p \leq T$ such that the trajectory of the state can be uniquely determined by the correspondence
\begin{equation}\label{eq:definition1}
[0,T] \ni t_j \rightarrow m_i(t_j) = Tr(Q_i \rho(t_j))
\end{equation}
for $i=1,...,r$ and $j=1,...,p$.
\end{defn}
The outcomes that we obtain from the measurements can be presented in a matrix form as
\begin{equation}\label{eq:matrix}
\left[\begin{matrix} m_1(t_1) & m_1(t_2) &\cdots & m_1 (t_p)\\ m_2 (t_1) & m_2(t_2) & ... & m_2(t_p)\\ \vdots & \vdots & \ddots & \vdots\\m_r (t_1) & m_r(t_2) & \cdots & m_r (t_p)\end{matrix}\right].
\end{equation}
The fundamental question that we formulate is: \textit{Can we reconstruct the initial density matrix} $\rho(0)$ \textit{ for a given master equation from the set of measurement results presented in \eqref{eq:matrix}? }

Other questions that arise in this approach concern: the minimal number of observables for a given master equation and their properties as well as the minimal number of time instants and their choice. The general conditions for observability have been determined and will be presented here as theorems and the proofs can be found in papers \cite{jam83,jam00,jam04}.

\begin{thm}\label{thm:1}
For a quantum system which evolution is given by the Kossakowski-Lindblad master equation of the form
\begin{equation}\label{eq:kossakowski}
\frac{ d \rho}{d t } = \mathbb{L} \rho,
\end{equation}
where the operator $\mathbb{L}$ is called the generator of evolution, there exists a number (denoted by $\eta$) which is called the index of cyclicity and is interpreted as the minimal number of observables required to reconstruct the density matrix. The index of cyclicity can be computed from the equality \cite{jam83}
\begin{equation}
\eta := \max \limits_{\lambda \in \sigma (\mathbb{L})} \{ dim Ker (\mathbb{L} - \lambda \mathbb{I})\},
\end{equation}
where $\sigma (\mathbb{L})$ denotes the spectrum of the generator of evolution (i.e. the set of all eigenvalues of $\mathbb{L}$).
\end{thm}

According to theorem 1 for every generator of evolution there always exists a set of $\eta$ observables such that the system is $(Q_1,...Q_{\eta})$-reconstructible. Moreover if the system is also $(Q_1,...Q_{\eta'})$-reconstructible, then $\eta' \geq \eta$. The index of cyclicity seems the most important factor when one is considering the usefulness of the stroboscopic approach to quantum tomography. This figure indicates how many distinct experimental setups one would have to prepare to reconstruct the initial density matrix in an experiment. The index of cyclicity is a natural number from the set $\{1, 2, \dots, N^2 -1\}$ (where $N= dim\mathcal{H}$) and the lower the number the more advantageous it is to employ the stroboscopic approach instead of the standard tomography. Moreover one can notice that the index of cyclicity can be understood as the greatest geometric multiplicity of eigenvalues of the generator of evolution. Thus, one can conclude that the index of cyclicity has physical interpretation, which is important from experimental point of view, but its value depends on the algebraic properties of the generator of evolution. Therefore, the question whether the stroboscopic tomography is worth employing or not depends primarily on the character of evolution of the quantum system. 

Another problem that we are interested in relates to the necessary condition that the observables $(Q_1,...Q_r)$ need to fulfill so that the system with dynamics given by \eqref{eq:kossakowski} will be $(Q_1,...Q_r)$-reconstructible. First, we introduce a denotation $ B(\mathcal{H})$ which shall relate to the vector space of all linear operators on $\mathcal{H}$. Then by $\langle A | B \rangle$ we shall denote the inner product in this space, which is defined as
\begin{equation}
\langle A | B \rangle = Tr(A^* B).
\end{equation}

Furthermore, one can notice that assuming the dynamics given by \eqref{eq:kossakowski} the formula for $\rho (t)$ at an arbitrary time instant can be expressed in terms of a semigroup
\begin{equation}
\rho (t) = exp (\mathbb{L} t) \rho(0) = \sum_{k=0} ^{\mu-1} \alpha_k (t) \mathbb{L}^k \rho(0),
\end{equation}
where $\mu$ stands for the degree of the minimal polynomial of $\mathbb{L}$.

This observation enables us to expand the formula for results of measurements (see \eqref{eq:definition1}) in the following way
\begin{equation}\label{eq:measure}
m_i (t_j) = \langle Q_i | \rho(t_j) \rangle = \sum_{k=0} ^{\mu-1} \alpha_k (t_j) \langle Q_i | \mathbb{L}^k \rho(0) \rangle =  \sum_{k=0} ^{\mu-1} \alpha_k (t_j) \langle (\mathbb{L}^*)^k Q_i | \rho(0) \rangle,
\end{equation}
where $\mathbb{L}^*$ is the dual operator to $\mathbb{L}$ or, in other words, $\mathbb{L}$ in the Heisenberg representation.

It can be proved that the functions $\alpha_k (t)$ are mutually linearly independent and can be computed from a system of differential equations \cite{jam04}. Therefore, the data provided by the experiment allows us to calculate the projections $\langle (\mathbb{L}^*)^k Q_i | \rho(0) \rangle$ for $ k = 0, 1, \cdots,\mu -1$ and $i= 1, 2, \cdots, r$. It can be observed that the initial state $\rho(0)$ (and consequently the trajectory $exp(\mathbb{L} t) \rho(0)$) can be uniquely determined if and only if the operators $(\mathbb{L}^*)^k Q_i$ span the vector space of all self-adjoint operators on $\mathcal{H}$. This space shall be denoted by $ B_*(\mathcal{H})$ and will be referred to as the Hilbert-Schmidt space. Now, if the evolution of the system is given by \eqref{eq:kossakowski} the conclusion can be presented as a formal theorem.

\begin{thm}
The quantum system is $(Q_1,...Q_r)$-reconstructible if and only if the operators $\{Q_1, \dots, Q_r\}$ fulfill the condition \cite{jam83,jam00}
\begin{equation}
\bigoplus \limits_{i=0}^r K_\mu (\mathbb{L},Q_i) = B_*(\mathcal{H}),
\end{equation}
where $\bigoplus$ denotes the Minkowski sum of subspaces, $\mu$ is the degree of the minimal polynomial of $\mathbb{L}$ and $K_\mu (\mathbb{L}, Q_i)$ denotes Krylov subspace, which is defined as
\begin{equation}
 K_\mu (\mathbb{L}, Q_i) := Span \{ Q_i, \mathbb{L}^* Q_i, (\mathbb{L}^*)^2 Q_i, ...,(\mathbb{L}^*)^{\mu-1} Q_i \}.
\end{equation}
\end{thm}
\begin{rem}
In the theorem 2 we denote by $Q_0$ an identity matrix of the appropriate dimension. One can notice that for any generator of evolution $\mathbb{L}$ we have $K_\mu (\mathbb{L}, \mathbb{I}) = \mathbb{I}$.
\end{rem}

When discussing the usefulness of the stroboscopic tomography, it can be observed that if one takes a hermitian operator $\tilde{Q}$ which belongs to the invariant subspace of the Heisenberg generator $\mathbb{L}^*$, then $ K_\mu (\mathbb{L}, \tilde{Q})= \tilde{Q}$. Therefore, multiple measurement of the same observable leads to projections of $\rho(0)$ into distinct operators only if the observable does not belong to the invariant subspace of $\mathbb{L}^*$. Thus if one considers the implementation of the stroboscopic tomography in an experiment, its effectiveness depends on whether one can measure such a quantity that the corresponding hermitian operator does not belong to the invariant subspace of the Heisenberg generator.

The last theorem which will be presented in this section gives the condition for the choice of time instants.

\begin{thm}
The determination of the initial state of the quantum system with evolution given by \eqref{eq:kossakowski} and which is $(Q_1,...Q_r)$-reconstructible is possible if the time instants $\{t_j\}_{j=1}^{\mu}$ satisfy the condition \cite{jam04}
\begin{equation}
det \left [ \alpha_k (t_j) \right ] \neq 0,
\end{equation}
where $k=0,1,..., \mu-1$.
In the above relation $\alpha_k(t_j)$ denotes the functions that appear in the polynomial representation of the semigroup $\Phi(t) = exp(\mathbb{L}t)$.
\end{thm}

Having summarized the most important general results concerning the stroboscopic tomography, we can proceed to analyzing specific examples. In the main part of this article there are three different decoherence models of 2-level quantum systems, to which the stroboscopic approach has been applied. Section 2 is devoted to the problem of dephasing, in which the stroboscopic approach allows us to give the concrete formula for the initial density matrix. In section 3 we discuss the usefulness of the stroboscopic approach in case of depolarization, which is another model of decoherence. Finally, in section 4 we tackle a more general problem, where the stroboscopic approach seems to have the greatest advantage. In that section we introduce a parametric-dependent family of Kraus operators for which the generator of evolution has no degenerate eigenvalues, i.e. in that case there exists one observable the measurement of which performed at three different instants is sufficient to reconstruct the initial density matrix.

\section{Quantum tomography model for dephasing}

In this section we are analyzing a decoherence model which is called dephasing. It is a model of two-level atoms subject to fluctuating external magnetic or laser fields. In geometric language it refers to shrinking of the Bloch ball in $x$ and $y$ directions, and $z$ being left intact.

The canonical Kraus operators have the following forms \cite{nielsen00}:
\begin{equation}\label{eq:1}
K_0 (t) = \sqrt{\frac{1+\kappa(t)}{2}} I \text{, } K_1 (t) = \sqrt{\frac{1-\kappa(t)}{2}} \sigma_3,
\end{equation}
where $\kappa(t)$ is a function which depends on time and can be expressed as $\kappa(t) = e^{-\gamma t}$, where $\gamma \in \mathbb{R}_+ $ is a dephasing parameter.

The collection of Kraus operators constitutes a completely positive and trace-preserving map (i.e. it is a quantum channel). Therefore, $\rho(t)$ at any time can be obtained from the formula
\begin{equation}\label{eq:map}
\rho(t) = \sum_{i=0} ^1 K_i(t) \rho(0) K_i^*(t) = \frac{1+ \kappa(t)}{2} \rho(0) + \frac{1-\kappa(t)}{2} \sigma_3 \rho (0) \sigma_3.
\end{equation}

Having the specific form of Kraus operators one can obtain the Kossakowski-Lindblad equation for evolution of such a system by differentiating the equation \eqref{eq:map}. It has the following form:
\begin{equation}\label{eq:2}
\frac{d \rho}{d t} = \frac{\gamma}{2} (\sigma_3 \rho \sigma_3 - \frac{1}{2} \{ \sigma_3 ^2 , \rho \} ).
\end{equation}
Here, since $\sigma_3$ is a self-adjoint operator, it can also be presented in terms of a double commutator
\begin{equation}\label{eq:3}
\frac{d \rho}{d t} = - \frac{\gamma}{4} [\sigma_3, [\sigma_3, \rho]].
\end{equation}

The explicit form of the generator of evolution can be obtained by using the relation from vectorization theory \cite{henderson81}
\begin{equation}\label{eq:4}
vec(ABC) = (C^T \otimes A) vec B,
\end{equation}
where it is assumed that  the matrices $A,B,C$ are selected in such a way that the matrix product $ABC$ is computable, i.e. the corresponding sizes of $A,B,C$ are $s_1\times s_2, s_2\times s_3$ and $s_3 \times s_4$, where $s_1,s_2,s_3,s_4 \in \mathbb{N}$.

Taking into account this property one can get the generator of evolution for this system in the matrix form
\begin{equation}\label{eq:5}
\mathbb{L} = -\gamma \left[\begin{matrix} 0 &0 &0 &0\\ 0 & 1 & 0 & 0\\0 & 0 & 1 & 0\\0 &0 &0 & 0\end{matrix}\right].
\end{equation}

Calculating the characteristic polynomial of this operator one obtains
\begin{equation}\label{eq:6}
det(\mathbb{L} - \lambda I ) = \lambda^2 (\lambda + \gamma)^2,
\end{equation}
from which one can observe that the index of cyclicity of the system with such a generator is equal $2$. According to the theorem \ref{thm:1} it means that there exist two observables the mean values of which enable us to reconstruct the initial density operator and, as a result, the whole trajectory of the state. We can instantly notice the benefit of the stroboscopic approach in comparison with the standard model - here we measure two different quantities instead of three.

Furthermore, it can be noticed that operator $\mathbb{L}$ fulfills an equality
\begin{equation}\label{eq:7}
\mathbb{L}^2 + \gamma \mathbb{L} =0,
\end{equation}
which means that $deg \text{ }\mu (\lambda, \mathbb{L}) = 2$, where by $\mu (\lambda, \mathbb{L})$ one should understand the minimal polynomial of $\mathbb{L}$. 

The observables $Q_1, Q_2$ that are needed to perform quantum tomography have to satisfy the necessary condition, which is
\begin{equation}\label{eq:8}
\bigoplus \limits_{i=0}^2 K_2 (\mathbb{L},Q_i) = B_*(\mathcal{H}),
\end{equation}
where $K_2 (\mathbb{L},Q_i$) denotes a Krylow subspace and can be rewritten as
\begin{equation}\label{eq:9}
K_2 (\mathbb{L},Q_i) = Span \{ Q_i, L^* Q_i\}.
\end{equation}
If we take  $Q_1 = \sigma_1$ and $Q_2 = \sigma_2 + \sigma_3$, it can be observed that $\mathbb{L}^* Q_1 = - \gamma Q_1 = -\gamma \sigma_1$ and $\mathbb{L}^* Q_2 = - \gamma \sigma_2$ . Then as $\mathbb{I},Q_1, Q_2$ and $\mathbb{L}^* Q_2$ are linearly independent the condition \eqref{eq:8} is fulfilled. 

Therefore, in order to reconstruct the density matrix $\rho (0)$ for the system which evolution is given by \eqref{eq:3} it is enough to measure $Q_1$ once and $Q_2$ twice for different time instants $t_1$ and $t_2$.

To obtain a specific formula for the density matrix we have to analyze the polynomial representation of the completely positive map $ \Phi (t) = exp (\mathbb{L}t)$. It has already been mentioned that $deg \text{ }\mu (\lambda, \mathbb{L}) = 2$, thus the polynomial representation takes form
\begin{equation}\label{eq:10}
exp ( \mathbb{L} t ) = \alpha_0 (t) \mathbb{I} + \alpha_1 (t) \mathbb{L}.
\end{equation}
Coefficients $\alpha_i (t)$ that appear in this equation can be easily computed because they satisfy a set of differential equations \cite{jam04}, which has been mentioned in the introductory section. One can easily get
\begin{equation}\label{eq:11}
\alpha_0 = 1 \text{ and } \alpha_1 (t) = \frac{1}{\gamma} (1 - e ^{- \gamma t}).
\end{equation}

Bearing in mind the general representation of a result of measurement (see equation \eqref{eq:measure}) we shall write the formula for $m_2 (t_1)$
\begin{equation}
m_2 (t_1) = \langle Q_2 | \rho(t_1) \rangle = \langle Q_2 | exp(\mathbb{L} t_1) \rho(0) \rangle = \sum_{k=0} ^1 \alpha_k (t_1) \langle (\mathbb{L}^*)^k Q_2 | \rho(0) \rangle.
\end{equation}
As the observable $Q_2$ is going to be measured twice, we obtain two similar equations for the two results. They can be combined into a matrix equation
\begin{equation}\label{eq:12}
\left[ \begin{matrix} m_2 (t_1) \\ m_2(t_2) \end{matrix} \right ] = \left[ \begin{matrix} 1 & \frac{1}{\gamma} (1 - e ^{- \gamma t_1})\\  1 & \frac{1}{\gamma} (1 - e ^{- \gamma t_2}) \end{matrix} \right ] \left[ \begin{matrix} \langle Q_2 | \rho (0) \rangle \\ \langle \mathbb{L}^* Q_2 | \rho (0) \rangle \end{matrix} \right],
\end{equation}
where $m_2(t_1)$ and $m_2(t_2)$ are results obtained from measurement, i.e. mean values of the observable $Q_2$ in two different time instants: $m_2(t_i) = Tr ( Q_2 \rho (t_i) )$. It is evident that if $ t_1 \neq t_2$,
\begin{equation}\label{13}
det \left[ \begin{matrix} 1 & \frac{1}{\gamma} (1 - e ^{- \gamma t_1})\\  1 & \frac{1}{\gamma} (1 - e ^{- \gamma t_2}) \end{matrix} \right ] \neq 0.
\end{equation}
Therefore, from equation \eqref{eq:12} we can calculate projections of $\rho (0)$ into the operators $Q_2$ and $\mathbb{L}^* Q_2$. We obtain the following results
\begin{equation}\label{eq:14}
\langle Q_2 | \rho (0) \rangle  = \frac{m_2(t_1) (1- e ^ {- \gamma t_2}) - m_2(t_2) (1 - e ^{- \gamma t_1})}{e^{- \gamma t_1} - e ^{- \gamma t_2}},
\end{equation}
\begin{equation}\label{eq:15}
\langle \mathbb{L}^* Q_2 | \rho (0) \rangle  = \frac{\gamma ( m_2 (t_2)-m_2 (t_1))}{e^{- \gamma t_1} - e ^{- \gamma t_2}}.
\end{equation}

For the observable $Q_1$ it is sufficient to write one equation and transform it in the appropriate way
\begin{equation}\label{eq:16}
\begin{aligned}
m_1 (t_1){} & = \langle Q_1 | \rho (0) \rangle + \frac{1}{\gamma} (1 - e ^{- \gamma t_1}) \langle \mathbb{L}^* Q_1 | \rho (0) \rangle = \\ &
= \langle Q_1 | \rho (0) \rangle + \frac{1}{\gamma} (1 - e ^{- \gamma t_1}) (-\gamma) \langle Q_1 | \rho (0) \rangle = \\ &
= (1+ ( e ^{- \gamma t_1}-1)) \langle Q_1 | \rho (0) \rangle = \\& = e^{- \gamma t_1} \langle Q_1 | \rho(0) \rangle,
\end{aligned}
\end{equation}
from which we get the projection of $\rho(0)$ into operator $Q_1$
\begin{equation}\label{eq:17}
\langle Q_1 | \rho (0) \rangle = m_1 (t_1)e ^{ \gamma t_1}.
\end{equation}

In order to obtain an explicit formula for the density matrix we will use the theorem that any two-dimensional density matrix can be expanded using the identity operator $\mathbb{I}$ and the traceless Pauli matrices $\{\sigma_1, \sigma_2, \sigma_3\}$. The decomposition takes the from
\begin{equation}\label{eq:18}
\rho (0) = \frac{1}{2} \left ( \mathbb{I} + \sum_{i=1}^3 s_i \sigma_i \right ),
\end{equation}
where $s_i = Tr(\sigma_i \rho(0)) = \langle \sigma_i | \rho (0) \rangle$.

We can notice that
\begin{equation}\label{eq:19}
\langle Q_2 | \rho (0) \rangle = \langle \sigma_2 + \sigma_3 | \rho(0) \rangle = \langle \sigma_2|\rho(0)\rangle + \langle \sigma_3 | \rho(0) \rangle,
\end{equation}
\begin{equation}\label{eq:20}
\langle \mathbb{L}^* Q_2|\rho(0) \rangle = - \gamma \langle \sigma_2 | \rho(0) \rangle.
\end{equation}
Taking these equations into account we get the projections we need to reconstruct the density matrix
\begin{equation}\label{eq:21}
\langle \sigma_2| \rho(0) \rangle = \frac{  m_2 (t_1)-m_2 (t_2)}{e^{- \gamma t_1} - e ^{- \gamma t_2}}
\end{equation}
and
\begin{equation}\label{eq:22}
\langle \sigma_3|\rho(0) \rangle = \frac{m_2(t_2) e^{- \gamma t_1} - m_2(t_1) e ^{- \gamma t_2}}{e^{- \gamma t_1} - e ^{- \gamma t_2}}
\end{equation}

Having found all the projections of $\rho (0)$ one can write an explicit formula for the density matrix
\begin{equation}\label{eq:23}
\rho (0) = \frac{1}{2} \left( \mathbb{I} + m_1 (t_1)e ^{ \gamma t_1} \sigma_1 + \frac{  m_2 (t_1)-m_2 (t_2)}{e^{- \gamma t_1} - e ^{- \gamma t_2}} \sigma_2 + \frac{m_2(t_2) e^{- \gamma t_1} - m_2(t_1) e ^{- \gamma t_2}}{e^{- \gamma t_1} - e ^{- \gamma t_2}} \sigma_3 \right ),
\end{equation}
which is the final result of this analysis.

In this section we were capable of creating a successful model of quantum tomography by using the stroboscopic approach. The advantage over the standard model is the fact that here one needs to measure only two different observables instead of three. The number of pairs of observables that fulfill the necessary conditions of observability for the system in question is infinite. Nevertheless, it has been shown for the selected two observables that it is possible to obtain the explicit formula for the initial density operator.

\section{An attempt to apply the stroboscopic tomography to depolarization}

Depolarization is another model of decoherence that is being analyzed in this paper. Geometrically speaking, this model refers to squeezing the Bloch ball uniformly in the three directions. The Kraus operators for depolarization have the forms \cite{nielsen00}
\begin{equation}\label{eq:24}
K_0 (t) = \sqrt{\frac{1+3\kappa (t)}{4}} \mathbb{I} \text{,  } \left \{ K_i (t) = \sqrt{\frac{1-\kappa (t)}{4}} \sigma_i \right \}_{i=1,2,3},
\end{equation}
where $\kappa (t)$ depends on time and can be expressed as $\kappa (t) = e^{-\gamma t}$, where $\gamma \in \mathbb{R}_+$ is called depolarizing parameter.

Having the collection of Kraus operators for this case one can calculate the derivative of $\rho$ analogously as in section 2. It leads to the evolution equation in the Kossakowski-Lindblad form. One can obtain
\begin{equation}\label{eq:25}
\frac{d \rho}{ d t } = \frac{\gamma}{4} \left (  \sum_{i=1} ^3 \sigma_i \rho \sigma_i - 3 \rho \right ).
\end{equation}
As in this case the operators $ \{ \sigma_i \}_{i=1,2,3}$ that govern the evolution are hermitian, the master equation can be represented as a sum of double commutators
\begin{equation}\label{eq:26}
\frac{d \rho}{ d t } = - \frac{\gamma}{8} \sum_{i=1}^3 [\sigma_i,[\sigma_i, \rho]].
\end{equation}

Applying the relation from vectorization theory \eqref{eq:4}, one can get the explicit form of the generator of evolution 
\begin{equation}\label{eq:27}
\mathbb{L} = \frac{\gamma}{4} \left ( \sigma_1 \otimes \sigma_1 + \sigma_2 ^T \otimes \sigma_2 + \sigma_3 \otimes \sigma_3 - 3 \mathbb{I}_4 \right ),
\end{equation}
which can also be presented in the matrix form
\begin{equation}\label{eq:28}
\mathbb{L} = -\frac{\gamma}{2} \left[\begin{matrix} 1 &0 &0 &-1\\ 0 & 2 & 0 & 0\\0 & 0 & 2 & 0\\-1 &0 &0 & 1\end{matrix}\right].
\end{equation}

Now one can find the eigenvalues of the operator $\mathbb{L}$
\begin{equation}\label{eq:29}
\lambda_1 =  -\gamma \text{ and } \lambda_2 = 0,
\end{equation}
and their corresponding multiplicities which are
\begin{equation}\label{eq:30}
n_1 = 3 \text{ and } n_2 = 1.
\end{equation}

One can quickly check that there are three linearly independent eigenvectors of $\mathbb{L}$ that correspond to the eigenvalue $\lambda_1$. It means that the index of cyclicity for the generator of evolution given by \eqref{eq:28} is equal $3$, which implies that we need $3$ different observables to perform quantum tomography on the system. Therefore, in case of depolarization stroboscopic approach to tomography has no advantage over the standard tomography model for a 2-level system.\\

\section{One-parametric non-degenerate family of Kraus operators}

In this section, before we introduce the main result, we shall assume two definitions. One concerns a collection of Kraus operators that can be associated with any completely positive map. The other one relates to a family of Kraus operators which should be understood as parametric-dependent and, therefore, more general case. To illustrate the difference between these two terms we shall first revise the theorem on completely positive maps \cite{bengtsson06}.
\begin{thm}
A linear map $\Phi (t):  B(\mathcal{H}) \rightarrow  B(\mathcal{H})$ is completely positive if and only if it is of the form
\begin{equation}
\Phi (t) (X) = \sum_{i=1} ^ \zeta K_i (t) X K_i ^* (t),
\end{equation}
where $ K_i (t) \in B(\mathcal{H})$.
\end{thm}

Now we can make a distinction between a collection of Kraus operators and a family of Kraus operators.
\begin{defn}
A set of operators $\{K_i (t)\}_{i=1}^{\zeta}$ that appear in the above decomposition of a completely positive map is called a collection of Kraus operators.
\end{defn}
\begin{defn}

If we have a completely positive map $\Phi(t, a)$, where $t$ stands for time and $a\in < a_{min}, a_{max}>$ is a real parameter that influences the structure of the generator of evolution $\mathbb{L}$ associated with the map, and we denote the sum representation of $\Phi(t, a)$ as
\begin{equation}
\Phi(t, a) (X) = \sum_{i=1}^{\zeta} K_i (t,a) X K_i^* (t,a),
\end{equation}
then we shall call the set of operators $\{K_i (t, a) \}_{i=1}^{\zeta}$ \textit{a one-parametric family of Kraus operators} on the interval $ < a_{min}, a_{max}>$.
\end{defn}
\begin{rem}

According to the definition of the one-parametric family of Kraus operators, the parameter has to influence the structure of the generator $\mathbb{L}$ and, consequently, the positive constant $\gamma$ from the two previous sections does not create a family of Kraus operators as the generator $\mathbb{L}$ is proportional do $\gamma$ (see \eqref{eq:5} and \eqref{eq:28}).
\end{rem}
\begin{rem}
In a similar way we could introduce the definition of \textit{k-parametric family of Kraus operators} for $k=2,3,\cdots$
\end{rem}

One can notice that according to these definitions a family of Kraus operators comprises an infinite number of collections of Kraus operators.

Now we can proceed to the main part of this section. As the third model of decoherence we are analyzing the case when the evolution of the open quantum system is given by a one-parametric family of Kraus operators
\begin{equation}\label{eq:31}
\begin{aligned}
{} & K_0 (t,a) = \sqrt{\frac{1+ 2 \kappa(t)}{3}} \mathbb{I} \text{,  } K_1(t,a) = \sqrt{\frac{a(1-\kappa(t))}{3}} \sigma_1 \\ & \text{ and } K_2(t,a) = \sqrt{\frac{(2-a)(1-\kappa(t))}{3}} \sigma_2,
\end{aligned}
\end{equation}
where $\kappa(t)$ depends on time according to $\kappa (t) = e^{-\gamma t}$, where $\gamma \in \mathbb{R}_+$ is a decoherence parameter, and $a$ is a real parameter that influences the structure of the generator of evolution.

Let us observe that it is a family of Kraus operators for $a \in \mathbb{R}$. Nevertheless, some constraints need to be found concerning $a$ because the family of Kraus operators from \eqref{eq:31} have to constitute a completely positive map which is strictly trace-preserving as the evolution should retain all the properties of the density operator. One can notice that the sufficient and necessary conditions for this are
\begin{equation}
a \geq 0 \text{ and } 2 - a \geq 0.
\end{equation}

It is easy to check that under these two condition the following equality holds
\begin{equation}\label{eq:34}
\sum_{i=0} ^2 K_i ^* (t,a) K_i (t,a) = \mathbb{I},
\end{equation}
thus any collection of Kraus operators taken from the family \eqref{eq:31} with $a \in <0;2> $ constitutes a quantum channel and can be treated as a model of decoherence.

In this section we propose the following theorem concerning the introduced family of Kraus operators.
\begin{thm}
If the evolution of an open quantum system is given by a collection of Kraus operators from the family defined in \eqref{eq:31} for any $a\in (0;2) $ \textbackslash $\{1\}$, then the eigenvalues of the generator of evolution $\mathbb{L}$ are non-degenerate.
\end{thm}
\begin{proof}

Kraus operators allow us to find the equation for evolution of the system, which takes the following form
\begin{equation}\label{eq:35}
\frac{d \rho}{d t} = - \frac{2}{3} \gamma \rho + \frac{a}{3} \gamma \sigma_1 \rho \sigma_1 + \frac{2-a}{3} \gamma \sigma_2 \rho \sigma_2.
\end{equation}

Using again the idea of vectorization \eqref{eq:4} the explicit form of the generator $\mathbb{L}$ can be written as
\begin{equation}\label{eq:36}
\mathbb{L} = \frac{\gamma}{3} \left ( - 2 \mathbb{I}_4 + a \sigma_1 \otimes \sigma_1 + (2-a) \sigma_2 ^T \otimes \sigma_2 \right ).
\end{equation}
Consequently its spectrum can be computed
\begin{equation}\label{eq:37}
\sigma( \mathbb{L} ) = \{ - \frac{4}{3} \gamma, 0, \frac{2}{3} (a-2)\gamma,-\frac{2}{3} a \gamma \}.
\end{equation}

Now it can be observed that for any $a\in (0;2) $ \textbackslash $\{1\}$ the spectrum of the generator $\mathbb{L}$ consists of four different eigenvalues, which completes the proof.
\end{proof}

The theorem presented in this section shows that there exists a one-parametric family of Kraus operators for which eigenvalues of the generator of evolution $\mathbb{L}$ are non-degenerate, which means that the index of cyclicity is equal 1. Henceforth a family that possesses this property shall be called a \textit{one-parametric non-degenerate family of Kraus operators}.

The result means that for any evolution given by \eqref{eq:31} with $a\in (0;2)$  \textbackslash $\{1\}$ there exists one observable the measurement of which performed at three different instants allows us to reconstruct the initial density matrix and, as a result, the trajectory of the state. From experimental point of view the stroboscopic approach seems to have a considerable advantage over the standard tomography as in this case an experimentalist needs to prepare only one kind of measurement and repeat it three times instead of performing three different measurements.

\section{Summary}
This paper gives a brief insight into the possible applications of the stroboscopic tomography to 2-level decoherence models. It has been shown that the usefulness of this approach differs depending on the kind of evolution. The most promising result is described in section 4. In that part it has been proved that there exists a one-parametric non-degenerate family of Kraus operators, which means that for an infinite number of generators of evolution the index of cyclicity is equal 1, i.e. for any $a\in (0;2) $ \textbackslash $\{1\}$ there exists one observable the measurement of which performed at three different time instants provides sufficient data to determine the trajectory of the state with dynamics given by \eqref{eq:31}. Further research into the problem of non-degenerate families of Kraus operators is planned for the foreseeable future.
\section*{Acknowledgement}
This research has been supported by grant No. DEC-2011/02/A/ST1/00208 of National Science Center of Poland.

\end{document}